\documentclass[11pt]{article}
\usepackage[margin = 1.2 in]{geometry}
\usepackage[utf8]{inputenc}
\usepackage{amsfonts}
\usepackage{amsmath}
\usepackage{amssymb}
\usepackage[english]{babel}
\usepackage{amsthm}
\usepackage{tikz}
\usepackage{float}
\usepackage{comment}
\usepackage[export]{adjustbox}
\usepackage{subcaption}
\usepackage{bbm}
\usepackage{enumerate}
\usepackage{hyperref}
\usepackage{cleveref}
\usepackage{natbib}
\setcitestyle{authoryear,open={(},close={)}}

\excludecomment{note}

\theoremstyle{definition}

\newtheorem{proposition}{Proposition}
\newtheorem{lemma}{Lemma}
\newtheorem{corollary}{Corollary}

\theoremstyle{remark}

\newcommand{\und}{\underline}

\title{Platform-Mediated Competition}
\author{Quitz\'{e} Valenzuela-Stookey }
\date{November 5, 2020}

\begin{document}
\maketitle


\begin{note}
\Huge\textcolor{blue}{WITH NOTES}   
\end{note}

\begin{center}
    \Large \textcolor{blue}{\href{https://northwestern.box.com/s/ggdovah9pxagfys6e38kecl31lb42j2p}{Click here for the latest version}}
\end{center}

\begin{abstract}
    Cross-group externalities and network effects in two-sided platform markets shape market structure and competition policy, and are the subject of extensive study. Less understood are the within-group externalities that arise when the platform designs many-to-many matchings: the value to agent $i$ of matching with agent $j$ may depend on the set of agents with which $j$ is matched. These effects are present in a wide range of settings in which firms compete for individuals' custom or attention. I characterize platform-optimal matchings in a general model of many-to-many matching with within-group externalities. I prove a set of comparative statics results for optimal matchings, and show how these can be used to analyze the welfare effects various changes, including vertical integration by the platform, horizontal mergers between firms on one side of the market, and changes in the platform's information structure. I then explore market structure and regulation in two in-depth applications. The first is monopolistic competition between firms on a retail platform such as Amazon. The second is a multi-channel video program distributor (MVPD) negotiating transfer fees with television channels and bundling these to sell to individuals.
\end{abstract}

Multi-sided platforms play a large and growing role in the economy. The core business of some of the world's largest companies, including Amazon, Alibaba, Facebook, and Google, fall into this category. The defining feature of platforms is that they match different agents. In the case of multi-sided platforms, agents can be divided into distinct groups, and matches occur across groups. Increasing data availability and developments of matching technology, both facilitated in many cases by the internet, have fueled the rise in multi-sided platform businesses. 

An important feature of the environments in which many platforms operate is the existence of cross-group externalities. A U.S. based retailer derives benefits from contracting with a Chinese manufacturer, and the manufacturer benefits from selling products to the retailer. In general, neither party appropriates the full surplus from their transaction. This would not be a problem if it was easy for the retailer to find an appropriate manufacturer, or vice-versa. However when searching for a partner is costly, the presence of externalities means that agents will generally under-invest in search. This explains the existence of a platform such as Alibaba. The platform facilitates matches, and internalizes the matching externalities through fees charged to agents on one or both sides of the market. 

In addition to cross-group externalities, platforms often take advantage of network effects. The benefit that one side of the market derives from the platform's services depends on the set of agents on the other side with whom they may be matched. Software developers would like to create programs for operating systems that have a large number of users, and users prefer operating systems which support many programs. Modern platforms generally engage in more sophisticated matching than simply granting all or nothing access to a network. Search engines prioritize certain results, and curate results based on user preferences. Cable providers allow customers to choose between many different packages consisting of different bundles of channels. By doing so, platforms fine-tune the network effects within the platform.  

Cross-group externalities and network effects have long been central to the literature on multi-sided platform design and regulation. However significantly less attention has been paid to \textit{with-group} externalities in multi-sided settings. Consider again Alibaba's role matching retailers and manufacturers. Cross-group network effects are present; retailers would like to search on a platform on which many manufacturers are available, and manufacturers would like access to the largest set of potential customers. However manufacturers are also competitors. Fixing the set of retailers using the platform, a manufacturer would prefer to compete with as few other manufacturers as possible. As in this example, within-group externalities often work in the opposite direction as network effects. 

Within-group externalities have important implications for the regulation of platforms. As noted above, platforms add value by internalizing cross-group externalities. Indeed, when network effects are large, it has been argued that the efficiency gains provided by large networks are sufficient to justify the existence of a monopolistic platform \citep{evans2003antitrust}. However the platform also internalizes, to some degree, the effects of competition between firms. That is, the platform internalizes within-group externalities. Thus the platform will have an incentive to reduce competition between firms beyond the socially optimal level. Within-group externalities also affect the welfare implications of vertical integration by the platform into the firm side of the market, as well has horizontal mergers between firms. 

This paper studies the implications of within-group externalities on the design and regulation of platforms. I consider a monopolistic platform whose role is to match each agent on one side of the market with a set of agents on the other (so called many-to-many matching). For example, Google's ad platform matches advertisers and websites. Each advertiser's ad may be shown on multiple websites, and websites display multiple ads. Matches are reciprocal; $i$ is matched with $j$ if and only if $j$ is matched with $i$. I will refer to one side of the market as firms and the other side as individuals, although the analysis applies equally well to a wide range of business-to-business activities. There are no within-side externalities for individuals; an individual's payoff depends on the set of firms with which they are matched. A firm's payoffs, however, depend not only on the set of individuals with which it is matched, but also on the set of firms with which each of these individuals is matched. In the web advertising example, the ``individuals'' are the websites, whose payoff depends only on which ads are being displayed on their site. The firms are the advertisers, who care not only about which sites display their ads, but also, potentially, about how many other adds are shown on the same sites (due perhaps to viewers' limited attention), and whether these ads are from their competitors.   

The model accommodates both vertical and horizontal differentiation between agents. An agent's vertical type relates to their own marginal value for better matches, whereas their horizontal type captures their attractiveness to the other side of the market. I show that when payoffs are suitably supermodular, optimal matches have a natural threshold structure, whereby agents are matched with those on the other side of the market who have high enough vertical types. Using this characterization, I prove a set of general comparative statics results on how the matching changes when payoffs shift in a way that makes some agents relatively more important. I show how in a broad class of problems, including those in which one set of agents is privately informed about their type, these comparative statics results can be used to perform welfare analyses. Within-group externalities are an essential component of competition analysis. For example, I show that when these effects are not present vertical mergers between the platform and firms will unambiguously benefit individuals in most cases.

I explore two applications in depth. I modify the canonical Dixit-Stiglitz model of monopolistic competition by giving a platform control over the set of firms that each individual has access to. This model applies to many settings; for example, Amazon mediating interactions between customers and vendors. I characterize the types of mergers between the platform and firms, and the types of information acquisitions by the platform, that make individuals better or worse off. I also study a multi-channel video program distributor (MVPD) negotiating transfer fees with television channels and bundling these to sell to individuals who are privately informed about their value for programming. I show that horizontal mergers between channels which are included in the basic cable package will make all individuals worse off, even if the merger does nothing but create cost synergies and has no direct anti-competitive effects. On the other hand, all individuals will be made better off if the merger is between channels only offered in the premium packages. Similar results obtain for vertical mergers between the MVPD and channels.

The general intuition for a number of the results can be illustrated by focusing on the effect of the platform acquiring a firm, denote by $j$. Assume that firms like matches with individuals, but dislike competition; their payoff from a match with individual $i$ is decreasing in the number of other firms that $i$ is matched with. Suppose that before the merger the platform was not able to capture the full surplus enjoyed by firm $j$. This will be the case when firms have bargaining power or private information. After the merger, on the other hand, the platform will internalize $j$'s surplus completely. This change has two effects on the matching structure. First, the platform will want to increase the payoff to firm $j$ by matching more individuals with $j$. This effect is analogous to ``eliminating markups'', as discussed in the classical vertical integration literature.\footnote{See \cite{riordan2005competitive}.} Second, the platform wants to reduce the competition faced by $j$ (analogous to ``raising rivals' costs''). To do this, the platform matches the individuals who are matched with $j$ with fewer additional firms. When there are no within-group externalities only the first effect is present. In this case the merger results in larger matching sets for all individuals, which I show implies higher payoffs for all individuals.\footnote{This conclusion holds even when individuals make monetary payments to the firm, in which case it is an implication of the envelope condition for the platform's revenue maximizing mechanism.} The second effect however, which is driven by competition between firms, has the effect of shrinking individual matching sets. The overall welfare effect depends on which of the two dominates. Using the characterization of optimal matches, I am able to identify cases in which the welfare effect is unambiguous. In general, an acquisition by the platform of a low-type firm will benefit individuals, and an acquisition of a high-type firm will harm individuals. 

The literature on competition policy for multi-sided platforms is extensive, and will not be summarized in full here. For a recent review see \cite{jullien2020economics}. Much of this literature focuses on competition between platforms, and ignores within-group externalities. Seminal theoretical contributions in this area were made by \cite{rochet2003platform}, \cite{caillaud2003chicken}, and \cite{armstrong2006competition}. My primary interest is on the implications of platform mediation for competition between firms. I therefore focus on a monopolistic platform, but introduce competition effects. \cite{pouyet2016vertical} study vertical integration in a model with competing platforms, but without within-side competition on a given platform. More recently, \cite{tan2020effects} incorporate within-group externalities into a model of competition between platforms. In \cite{tan2020effects} platforms have a membership structure, and do not engage in more sophisticated matching, whereas in the current paper I allow the platform to flexibly design the matching and transfer schedule. The authors show that increasing the number of platforms can adversely affect consumer welfare. More closely related to the current paper is \cite{de2014integration}, which studies a search engine matching users and publishers. Publishers in turn make revenue by selling space to an advertiser. The search engine also sells ad space directly to the advertiser. This generates competition between the publishers' and search engine's web pages. While the structure of the model is quite different than that considered here, the authors identify broadly similar effects of vertical integration. Integration by search engine into publishing can induce bias into search results, analogous to the bias in matching sets that I find. However, as in my setting, there is a countervailing effect; integration may also reduce the quantity of ads seen by users, which benefits them. The net effect on user welfare remains ambiguous. 

The key ingredients to the model are \textit{i}) a platform that controls the interactions between agents on different sides of the market via many-to-many matchings, and \textit{ii}) within-side competition effects on one side. While these features appear separately in the literature, they have not been previously considered together. This paper builds on the matching design and price discrimination literature. The model of the platform is similar to that of \cite{gomes2016many}, who also consider the design of many-to-many matchings by a platform. However their model does not allow for within-side competition. As in \cite{gomes2016many}, the platform in this paper may engage in price discrimination by offering a menu of matching sets and fees to each side of the market. The platform can flexibly design both the matching sets and fees. This is in contrast to the literature on two-sided markets, in which platforms sell access to a single network, or to different mutually exclusive networks (see \cite{rysman2009economics} for a survey and \cite{weyl2010price}, \cite{white2016insulated} for more recent contributions). 

The remainder of this paper is organized as follows. Section \ref{sec:model} presents the basic model, characterizes optimal matchings, and discusses the main comparative statics results. Section \ref{sec:extension} discusses an extension of the model, which is useful in applications. Section \ref{sec:applications} presents the two applications mentioned earlier.

\section{The basic model}\label{sec:model}

The model generalizes that of \cite{gomes2016many}. There is a unit mass of individuals (side $I$) and a set $\mathcal{F}$ of firms (side F). Depending on the setting, I will consider either finite $\mathcal{F}$ or $\mathcal{F} = [0,1]$. In what follows $\lambda$ will be used to denote either Lebesgue measure, in the case of a continuum of firms, or the measure placing mass $1$ on each firm when $\mathcal{F}$ is finite. Competition effects are present only on the firm side, and take a form that will be specified below.\footnote{It is also interesting to consider markets with competition effects on both sides. For now I will focus on the one-sided case because most of the applications that I have in mind are of this form.} Matchings are reciprocal: individual $i$ is matched with firm $j$ iff $j$ is matched with $i$. 

\subsection{Payoffs}

I present here the baseline model. Alternative payoff structures are explored in the extensions in Section \ref{sec:extension}. Agent $\ell \in [0,1]$ on side $K \in \{F,I\}$ is characterized by a \textit{vertical type} $v_K^{\ell}$ and a \textit{horizontal type} $\sigma_K^{\ell} \geq 0$. The platform's objective function will be of the form
\begin{equation}\label{eq:objective}
\int_F U^F(v_F^j, |s_F(j)|_{S_I}) d\lambda(j) + \int_I U^I(v_I^i, |s_I(i)|) d \lambda(i).
\end{equation}
The components of (\ref{eq:objective}) will be discussed below. What is important to note at present is that the platform's payoff depends on the sum of some aggregate payoffs coming from the firm side and the individual side. If the platform's objective is utilitarian welfare maximization then $U^F$ and $U^I$ will correspond to the utilities of firms and individuals respectively. However there are many other objectives of the form given in (\ref{eq:objective}). One such setting, in which agents are privately informed about their vertical type, will be discussed in Section \ref{sec:privateinfo}. Further examples will be illustrated in the applications of Section \ref{sec:applications}. For simplicity, in what follows I will refer to $U^F$ and $U^I$ as firm and individual payoffs respectively. Later, when discussing individual welfare I will be careful to differentiate between the true utilities of agents and the payoffs that are relevant for the firm's objective in (\ref{eq:objective}). 

The vertical type $v_K^{\ell}$ determines the value that $\ell$ attaches to matches with agents on the other side of the market, while the horizontal type $\sigma_K^{\ell}$, which I also refer to as \textit{salience}, represents how important $\ell$ is to agents on the other side. The be precise, for an individual $i$ on side $I$ the payoff of being matched with a set $s_I(i) \subseteq [0,1]$ of firms is
\begin{equation*}
    U^I(v_I^i, |s_I(i)|)
\end{equation*}
where $|s_I(i)|$ is the salience weighted size of the set $s_I(i)$, given by
\begin{equation*}
    |s_I(i)| = \int\limits_{j \in s_I(i)} \sigma_F^j d \lambda(j).
\end{equation*}

Payoffs on the firms side are similar, but account for competition effects. This means that the payoff of a firm depends on the entire matching, not just their own matching set. The payoff to firm $j$ with matching set $s_F(j)$ when individuals have matchings $S_I = \{(s_I(i),i)\}_{i\in[0,1]}$ is given by 
\begin{equation*}
    U^F(v_F^j, |s_F(j)|_{S_I})
\end{equation*}
where
\begin{equation*}
    |s_F(j)|_{S_I} = \int\limits_{i \in s_F(j)} h(|s_I(i)|, \sigma_I^i, \sigma_F^j,) d \lambda(i),
\end{equation*}
and $h$ is non-negative. The function $h$ captures competition effects. It can be thought of as depending on the exogenous ($\sigma_I^i$) and endogenous ($|s_I(i)|$) components of individual salience.\footnote{$h$ need not be a function of $|s_I|$. The general results presented below apply as long as $h$ is continuous (in an appropriate sense) in $s_I$ and independent of the vertical types $v^j_F$ for $j \in s_I$. For example, $h$ could be a function of $\lambda(s_b)$ rather than $|s_I|$.} In most applications $h$ will be decreasing; individuals are less valuable to firms if they are matched with many other firms. However I will not assume that this is always the case.

In the model described above, neither firms nor individuals care about the vertical types of the agents they are matched with. The interpretation is the an agent's vertical type is a private taste parameter that describes their value for matches. I will discuss settings in which this assumption is natural. On the other hand, in some situations we may derive the vertical type from characteristics of the agent, for example the cost function of a firm that is also choosing prices, that may be relevant for agents on the other side. In Section \ref{sec:monopcomp} I explore payoffs of this form. The main qualitative conclusions of the model will be the same in both cases. 

\subsection{Optimal matchings}

For the time being, assume that there is no horizontal differentiation; $\sigma_I^i = \sigma_I^k$ for all $i,k$ and $\sigma_F^j = \sigma_F^l$ for all $j,l$. If there is horizontal differentiation, the results of this section will hold conditional on horizontal types. 

\vspace{3mm}
\noindent \textbf{Supermodularity.} Let $v''>v'$, $x'' > x'$. Then $U^K(v'', x'') + U^K(v',x') \geq U^K(v'', x') + U^K(v',x'')$ for $K \in \{F,I\}$.
\vspace{3mm}

\noindent \textbf{Order-Supermodularity.} There exists a complete order on the agents on side $K$ such that Supermodularity holds with respect to this order. That is, if type $\hat{v}$ is higher in this order than type $\tilde{v}$ (not necessarily $\hat{v} > \tilde{v}$) and $x'' \geq x'$, then $U^K(\hat{v}, x'') + U^K(\tilde{v},x') \geq U^K(\hat{v}, x') + U^K(\tilde{v},x'')$ for $K \in \{F,I\}$.
\vspace{3mm}

I will refer to the order specified in the definition of Order-Supermodularity as the supermodular order. Clearly, Supermodularity is a special case of Order-Supermodularity in which the order is given by type. 

\begin{lemma}\label{lem:monotonicity}
Under Order-Supermodualrity, optimal matchings are monotone in the supermodular order: higher type individuals receive larger matching sets and higher type firms receive higher quality matching sets. 
\end{lemma}
\begin{proof}
Without loss, let the order be given by type. Consider first individual side monotonicity. Let $v^I_j \geq v^I_k$ and suppose $|s_I(k)| \geq |s_I(j)|$. Then switch the matching sets. By Supermodularity, payoffs on the individual side increase. Moreover, payoffs on the firm side are unchanged, since any firm that was matched with $j$ is now matched with firm $i$ that has the same matching set that $j$ had. The same switching argument works on the firm side. 
\end{proof}

For the remainder of this section assume that Supermodularity holds. All results extend immediately to other supermodular orders. 

I now turn to establishing the aforementioned threshold structure of matching sets. One might be tempted to apply a similar proof as that of Lemma \ref{lem:monotonicity}. To see why this does not work, suppose $g_F$ is increasing and concave. Fix an individual, and suppose that they are matched with a firm with type $v'$ but not with a type $v''$ firm,  where $v'' > v'$. Simply dropping the low type firm from and adding the high type firm to the individual's set clearly does not change the individual's payoff, but will not necessarily improve payoffs on the firm side. This is because despite having a higher vertical type, the $v''$ firm will also have a higher quality matching set, by Lemma \ref{lem:monotonicity}. Then by concavity this means precisely that the marginal change in match quality for the $v''$ firm is lower than for the $v'$ firm. The proof of the following proposition modifies the switching argument to accommodate this case.

\begin{proposition}\label{prop1.1}
For an individual with type $v_I$ there is a threshold $v^*(v_I)$ such that the individual is matched with a firm if and only if the firm's type is above $v^*_F(v_I)$.
\end{proposition}
\begin{proof}
First I claim that the optimal matching should be characterized by a threshold in the marginal firm utility $U^F_2(v_F^j, |s_F(j)|_{S_I})$ (or the discrete analog). If this was not the case then switching out low marginal utility firms for high marginal utility firms does not change the size of the individual's matching set, and thus has no effect on the individual's payoffs or their endogenous salience. Moreover it increases firm side payoffs.\footnote{I am using here the fact that monotonicity does not bid. If it did then I would have to guarantee that this switch does not violate monotonicity. This could probably be dealt with, but I don't need to. }

Firm marginal utilities are of course endogenous objects. The result will follow if the optimal matching the firm marginal utilities are increasing in $v_F$, which is what I know show.

Let $\{S_F, S_I \}$ be the optimal matching. A necessary condition for $\{S_F, S_I \}$ to be optimal is that each individual's match be characterized by a threshold in the marginal firm utilities induced by $\{S_F, S_I \}$, as discussed above. Suppose that $U^F_2(v_F^j, |s_F(j)|_{S_I})$ is not increasing in firm type, that is, there are firms $i,j$ with $v^j > v^k$ such that
\begin{equation}\label{eq2.1}
    U^F_2(v_F^k, |s_F(k)|_{S_I}) > U^F_2(v_F^j, |s_F(j)|_{S_I}).
\end{equation}
For this to hold there must be a positive measure of individuals for whom the marginal utility threshold defining their matching set falls strictly above $U^F_2(v_F^j, |s_F(j)|_{S_I})$ and weakly below $ U^F_2(v_F^k, |s_F(k)|_{S_I})$, (otherwise $|s_F(v')|_{s_I}) = |s_F(v'')|_{s_I}$ and so (\ref{eq2.1}) does not hold). Then $s_F(v'') \subsetneq s_F(v')$. Since $h \geq 0$, this implies $|s_F(v')|_{s_I} > |s_F(v'')|_{s_I}$, contradicting monotonicity from Lemma \ref{lem:monotonicity}.\footnote{If $g_F$ is concave then we need not appeal to monotonicity, $|s_F(v')|_{s_I} > |s_F(v'')|_{s_I}$ contradicts (\ref{eq2.1}). The argument can be modified to accommodate for negative values of $h$.}
\end{proof}

\begin{corollary}\label{cor2.1}
The threshold $v_F^*(v_I)$ is decreasing. 
\end{corollary}
\begin{proof}
Immediate from Lemma \ref{lem:monotonicity} and Proposition \ref{prop1.1}
\end{proof}

\begin{corollary}\label{cor2.2}
Firm matchings are characterized by a threshold $v_I^*(v_F)$. Moreover $v_I^*(v_F)$ is decreasing. 
\end{corollary}
\begin{proof}
Immediate from Corollary \ref{cor2.1}.
\end{proof}

\subsection{Optimal matchings with horizontal differentiation}

The previous section established that optimal matchings have a threshold structure when there is no horizontal differentiation. With horizontal differentiation, the same structure continues to hold conditional on horizontal types; each individual $i$'s matching set will be characterized by a threshold function $v^*_I(\sigma_F, i)$ such that the individual matches with a firm of type $v_F$ and salience $\sigma_F$ if and only if $v_F \geq v^*_I(\sigma_F,i)$. Similarly for firms. In this section I will be interested in the features of optimal threshold functions. In particular, I will identify when this functions are increasing/decreasing. 

The platform's problem simplifies greatly when $U_F(v, \cdot)$ is affine. In this case the firm's payoffs can be written as $U_F(v,x) =a^F(v) + b^F(v) \cdot x$. If supermodularity holds then we can make a change of variables so that without loss of generality $b^F(v) = v$. Then the platform's objective function can be written as
\begin{equation*}
\int_I \left(U^I(v_i^I, |s_I(i)|) +  \int_{s_I(i)} v^F_j \cdot h(|s_I(i)|, \sigma^I_i, \sigma^F_j) \ d\lambda(j) \right) d\lambda(i).
\end{equation*} 
This integral can be maximized point-wise for each individual. This is true even when match qualities are constrained to be monotone in agent type, as will be the case when types are private infromation, since Lemma \ref{lem:monotonicity} tells us that these constraints will not bind. For simplicity, assume all primitive functions are differentiable. Consider the objective 
\begin{equation*}
U^I(v_i^I, |s_I(i)|) +  \int_{s_I(i)} v^F_j\cdot h(|s_I(i)|, \sigma^I_i, \sigma^F_j) \ d\lambda(j).
\end{equation*}
The marginal effect of adding a firm of type $v^F$ and salience $\sigma^F$ to the matching set of individual $i$ is
\begin{equation}\label{eq:horizontal_foc}
v^F \cdot h(|s_I(i)|, \sigma^I_i, \sigma^F) + \sigma^F \cdot\left( U^I_2(v_i^I, |s_I(i)|) + \int_{s_I(i)} v^F_j \cdot h_1(|s_I(i)|, \sigma^I_i, \sigma^F_j) \ d\lambda(j)\right) .
\end{equation}

The function $v^*_F(\sigma_F, i)$ will be decreasing (increasing) if for $\sigma_F'' > \sigma_F'$ the following single crossing property holds: if the marginal benefit in (\ref{eq:horizontal_foc}) is positive (negative) for $\sigma_F'$ then it is positive (negative) for $\sigma_F''$. Assume $h$ is decreasing it its first and third arguments (If $h$ is not decreasing in $\sigma^F$ then the threshold functions may be non-monotone). The term in parentheses in (\ref{eq:horizontal_foc}), $U^I_2(v_i^I, |s_I(i)|) + \int_{s_I(i)} v^F_j \cdot h_1(|s_I(i)|, \sigma^I_i, \sigma^F_j) \ d\lambda(j)$, is the key determinant of the slope of $v^*_F(\sigma_F, i)$. The first part of this expression is the marginal benefit to the individual of increasing the size of their matching set. The second part is the inframarginal cost to all firms matched with this individual of increasing the size of the matching set. If the sum of these two is positive, it means that the benefit to the individual outweighs the externality imposed on firms. But then this individual should be matched with all firms that have positive values $v^F)j$. Moreover, the individual should be matched with firms with negative values only if $\sigma^F_j$ is large enough, so that the benefit the individual is sufficient to outweigh both the cost to the new firm and the inframarginal cost to all other firms. If on the other hand $U^I_2(v_i^I, |s_I(i)|) + \int_{s_I(i)} v^F_j \cdot h_1(|s_I(i)|, \sigma^I_i, \sigma^F_j) \ d\lambda(j) < 0$ then the individual should never match with negative value firms, and should only match with positive value firms if their salience is not too high. These observations are summarized in the following Lemma.

\begin{lemma}\label{lem:horiz_structure}
Assume $h$ is decreasing in its first and third arguments (match size and firm salience).\footnote{Just assuming that $h$ is decreasing in match size, we can conclude that if $U^I_2(v_i^I, |s_I(i)|) + \int_{s_I(i)} v^F_j \cdot h_1(|s_I(i)|, \sigma^I_i, \sigma^F_j) \ d\lambda(j) \geq 0$ then the individual will be matched with all positive value firms.} Then the matching sets of individuals have the following structure:
\begin{itemize}
	\item There exists a threshold $v^{**}$ such that an individual $i$'s matching set contains firms with types $v^F_j < 0$ if and only if $v^I_i > v^{**}$. 
	\item The individual's threshold function $v^*_F(\sigma_F,i)$ is downward sloping if $v^I_i > v^{**}$ and upward sloping otherwise.
\end{itemize}	
\end{lemma}

Lemma \ref{lem:horiz_structure} has some interesting implications. If any firms have positive values then high-value, low-salience firms will be matched with the largest set of individuals (in the set inclusion sense). If no firms have positive values then the largest matching sets will instead go to high-value, high-salience firms. 

Suppose that firms vertical types are their private information, while horizontal types are known to the platform. For example, the vertical type may reflect a firm's marginal cost, while its horizontal type is the attractiveness of its product to consumers. Below, I will discuss more extensively the platform's problem when types are private information. Suppose that the platform make monetary transfers with firms, and that firms payoffs are quasi-linear in money. In this context it is easy to see, from the usual \cite{myerson1981optimal} argument, that the objective of a revenue maximizing platform will look the same as $\ref{eq:objective}$, except that the firms' vertical types $v^F_j$ will be replaced with their ``virtual values'' $\varphi(v^F_j, \sigma^F_j) = v^F_j - \frac{1 - Q_F(v^F_j| \sigma^F_j)}{q_F(v^F_j| \sigma^F_j)}$. Virtual values may be negative even if all true values are positive. As a result, the second best matching, in which firm types must be elicited, and the first best matching in which they are known may be qualitatively different: in the first best, conditional on values, low-salience firms will always be matched with a superset of the individuals matched with high-salience firms. In the second best the reverse may hold: for low enough values, high-salience firms are matched with a superset of the individuals that low-salience firms with the same value are matched with.

\subsection{Comparative statics}

Throughout this section, assume that there is no horizontal differentiation on either side. I will be interested in how the optimal matching changes when some of the firms become ``more prominent'', in a sense I will make precise. Intuitively, there are two effects of a platform placing greater weight on the payoffs of firm $j$. For one, it should increase the size firm $j$'s' matching set. This effect is easy to see. On the other hand, it should change the other firms' matching sets to increase the value of firm $j$'s matching. If $h$ is decreasing this means reducing the size of the other firms' matching sets. How these two effects interact is not obvious in general. However there will be cases in which the effect on the matching structure can be identified. I explore such cases in this section. 

By ``more prominent'' I mean that the marginal value of increasing the quality of a firm's matching increases. 

\vspace{3mm}
\noindent \textbf{Increasing differences change.}  A change in the payoffs of firm $j$ from $U^F$ to $\hat{U}^F$ are an \textit{increasing differences change} if for $x'' > x'$, $\hat{U}^F(v^F_j, x'') - \hat{U}^F(v^F_j, x') \geq U^F(v^F_j, x'') - U^F(v^F_j, x')$.

\begin{lemma}\label{lem6.1}
The quality of firm $j$'s matching set increases following an increasing differences change in its payoffs. 
\end{lemma}
\begin{proof}
This follows from the resulting single-crossing property of the objective.
\end{proof}

To further understand the changes in the matching, is important to first clarify the necessary conditions for optimality of a matching. Assume that there are $N < \infty$ firms and a continuum of individuals. The argument can be easily modified for the case of finitely many individuals. For notational simplicity, assume that Supermodularity holds (all results go through under Order-supermodularity). Label the firms in order of types, with firm 1 being the highest type and firm $N$ the lowest.

If there is no horizontal differentiation $h(n)$ is the endogenous salience of an individual matched with $n$ firms. Let $v_I^*(j)$ be the threshold type for firm $j$'s matching set. If firm $j$ marginally increases the size of its matching set by lowering the cut-off it benefits from a larger matching set, but affects the endogenous salience of the newly added individuals, who are included in the matching sets of all higher types. Suppose that firm $j$ adds a marginal individual to its matching set. Assume that there is not another firm with the exact same matching set as $j$, and that $U^I$ is continuous in its first argument. Using the fact that the optimal matching has the threshold structure described in Proposition \ref{prop1.1}, the FOC for such a change is given by\footnote{If another firm has the same matching set we just need to modify the term $[h(j-1) - h(j)]$ in equation (\ref{eq:foc}).}
\begin{equation}\label{eq:foc}
\begin{split}
    U^F_2(v^F_j, |s_F(j)|_{S_I}) \cdot h(j) - &\sum_{k = 1}^{j-1} U^F_2(v^F_k, |s_F(k)|_{S_I}) \cdot [h(j-1) - h(j)] \\
    &+ U^I(v_I^*(j), j) - U^I(v_I^*(j), j-1)  \geq 0
\end{split}
\end{equation}
with equality if $v^*_I(j)$ is interior.

The first term in (\ref{eq:foc}) is the marginal benefit to firm $j$, the second is the inframarginal cost to all firms matched with the newly added individual, and the last is the direct benefit to the new individual. 

\begin{proposition}\label{prop:comp_stat1}
Suppose there is an increasing differences change in the payoffs of firm $k$, and that the supermodular order remains unchanged. If $U^F(v, \cdot)$ is concave for all $v$ and $h$ is decreasing then all firms $j > k$ receive smaller matching sets and firm $k$ receives a larger matching set. 
\end{proposition}
\begin{proof}
Throughout this proof $s$ and $S$ will refer to the original matchings and $\tilde{s}$, $\tilde{S}$ will refer to the new matching. 

If $k=N$ then the proposition follows immediately from Lemma \ref{lem6.1}, so assume $k > 1$.  

\textit{Claim 1. There cannot exists an $\ell$ such that firms $N, N-1, \dots,\ell$ receive larger matching sets after the change.}
First, suppose that not all firms receive larger matching sets. Let $j$ be the lowest value firm that receives a strictly smaller matching set ( Lemma \ref{lem6.1} implies that $j \neq k$). Since $j$ receives a smaller matching set $\tilde{v}^*_I(j)$ must be larger than $v^*_I(j)$. Then Supermodularity implies that
\begin{equation*}
 U^I(\tilde{v}_I^*(j), j) - U^I(\tilde{v}_I^*(j), j-1) > U^I(v_I^*(j), j) - U^I(v_I^*(j), j-1).
\end{equation*}
 Since all lower value firms receive larger matching sets and $h$ is decreasing, firm $j$'s matching set is of lower quality. By concavity, $U^F_2(v^F_j, |\tilde{s}(j)|_{\tilde{S}_I}) > U^F_2(v^F_j, |s(j)|_{S_I})$. Then the FOC for $j$ holds only if
\begin{equation*}
    \sum_{t = 1}^{j-1} U^F_2(v^F_t, |\tilde{s}_F(t)|_{\tilde{S}_I}) \cdot [h(j-1) - h(j)]
\end{equation*}
is larger as well. But then we have that 
\begin{equation}\label{eq6.1}
    \sum_{t = 1}^j U^F_2(v^F_t, |\tilde{s}_F(t)|_{\tilde{S}_I}) \cdot [h(j-1) - h(j)] > \sum_{t = 1}^j U^F_2(v^F_t, |s_F(t)|_{S_I}) \cdot [h(j-1) - h(j)] 
\end{equation}
is also larger. Since firm $j+1$ received a larger matching set $\tilde{v}^*_I(j+1)$ is smaller. Then Supermodularity implies that 
\begin{equation*}
 U^I(\tilde{v}_I^*(j+1), j+1) - U^I(\tilde{v}_I^*(j+1), j) < U^I(v_I^*(j+1), j+1) - U^I(v_I^*(j), j).
\end{equation*}
Then the FOC for $j+1$ and (\ref{eq6.1}) imply that $U^F_2(v^F_{j+1}, |\tilde{s}_F(j)|_{\tilde{S}_I})$ is larger (or $\hat{U}^F_2(v^F_{j+1}, |\tilde{s}_F(j)|_{\tilde{S}_I})$ if $j+1 = k$). Proceeding in this way we conclude that $U^F_2(v^F_{N}, |\tilde{s}_F(N)|_{\tilde{S}_I})$ must be larger. But if firm $N$ received a larger matching set then concavity implies that $U^F_2(v^F_{N}, |\tilde{s}_F(N)|_{\tilde{S}_I})$ is strictly smaller after the merger (since $k > 1$ by assumption). Thus we have a contradiction. 

Now suppose all firms receive larger matching sets. If firm $1$ receives a larger matching set then (\ref{eq:foc}) implies that its marginal value $U^F_2(v^F_1, |\tilde{s}(1)|_{\tilde{S}_I})$ must be higher. Since firm $2$ receives a larger matching set $\tilde{v}^*_I(2)$ must be smaller. Then (\ref{eq:foc}) implies that $U^F_2(v^F_2, |\tilde{s}(2)|_{\tilde{S}_I})$ must be larger. Proceeding in this way we conclude that $U^F_2(v^F_N, |\tilde{s}(N)|_{\tilde{S}_I})$ is larger after the change, which we have already noted is a contradiction.

\textit{Claim 2. No firm $j > k$ can receive a larger matching set after the change.}
Let $j$ be the lowest value firm that receives a larger matching set, and suppose $j > k$. By Claim 1, $j < N$. Since $j$ receives a larger matching set $\tilde{v}^*_I(j)$ is smaller, so Supermodularity implies that $U^I(\tilde{v}_I^*(j), j) - U^I(\tilde{v}_I^*(j), j-1)$ is smaller. Since all firms $m > j$ receive smaller matching sets, concavity implies that $U^F_2(v^F_j, |\tilde{s}(j)|_{\tilde{S}_I})$ is smaller. Then the FOC for $j$ implies that 
\begin{equation*}
    \sum_{t = 1}^{j-1} U^F_2(v^F_t, |\tilde{s}_F(t)|_{\tilde{S}_I}) \cdot [h(j-1) - h(j)]
\end{equation*}
is smaller. But then we have that 
\begin{equation}
    \sum_{t = 1}^j U^F_2(v^F_t, |\tilde{s}_F(t)|_{\tilde{S}_I}) \cdot [h(j-1) - h(j)] < \sum_{t = 1}^j U^F_2(v^F_t, |s_F(t)|_{S_I}) \cdot [h(j-1) - h(j)] 
\end{equation}
Since $j+1$ receives a smaller matching set $\tilde{v}^*_I(j+1)$ is larger. But then the FOC for $j+1$ implies that $U^F_2(v^F_{j+1}, |\tilde{s}(j+1)|_{\tilde{S}_I})$ is smaller. Proceeding in this way we conclude that $U^F_2(v^F_N, |\tilde{s}(N)|_{\tilde{S}_I})$ is smaller. But this contradicts the assumption that $N$ receives a smaller (and thus worse) matching set, given concavity. 

\textit{Claim 3. Firm $k$ receives a larger matching set.} Suppose $k$ receives a smaller set. Suppose some firm with a higher value than $k$ receives a larger set, and let $j$ be the lowest value such firm. Then using the same proof as in Claim 2 we can arrive at a contradiction, so no firm can receive a larger set. If firm 1 receives a smaller matching set then the FOC for 1 implies that $U^F_2(v^F_1, |\tilde{s}(1)|_{\tilde{S}_I})$ (or $U^F_2(v^F_1, |\tilde{s}(1)|_{\tilde{S}_I})$ if $k=1$) must be smaller. Then again we can use the proof of Claim 2 to arrive at a contradiction.
\end{proof}

The assumption of concavity of $U_F(v,\cdot)$ is used throughout the proof, but it is not necessary. This can be seen most easily by appealing to continuity of the objective function. The strict version of the comparative statics result of Proposition \ref{prop:comp_stat1} holds when $U_F(v,\cdot)$ is affine, and given appropriate continuity assumptions Berge's theorem implies that it will continue to hold if $U_F(v,\cdot)$ is perturbed slightly to be strictly convex. Thus it is not clear exactly what the role of concavity is in the result. To try to build some intuition, consider an increasing differences change to the payoffs of firm 1. At the new optimal matching, the sum of payoffs for all other firms must be smaller (otherwise it would have been worthwhile to change their matchings before the change in firm 1's payoffs). Roughly speaking, concavity of $U_F(v,\cdot)$ implies that it is optimal to spread the reduction in payoffs across all these firms. This intuition is incomplete, and it would be valuable to know how far convexity of $U_F(v,\cdot)$ can be pushed. 

Let $v^*_I(k)$ be the threshold individual type defining firm $k$'s matching set. If all firms with lower types than $k$ receive smaller sets then all individuals with types higher then $v^*_I(k)$ receive smaller matching sets. 
{}
\begin{corollary}\label{cor1.1}
Under the conditions of Proposition \ref{prop1.1}, all individuals with types higher than $v^*_I(k)$ receive (weakly) smaller matching sets. 
\end{corollary}

If an individual's utility is given entirely by $U^I$, and is increasing in match size, then Corollary \ref{cor1.1} implies that all individuals with types above $v^*_I(k)$ are weakly worse off following the change in $k$'s payoffs. Even if this is not the case, for example if there are transfers between individuals and the platform, it may be possible to determine the welfare effects of such a change. I return to this question in the next section. 

It is easy to see from the proof of Proposition \ref{prop:comp_stat1} that the result can be extended to increasing changes in the payoffs of multiple firms. 

\begin{lemma}\label{lem:compstat_multi}
Under the conditions of Proposition \ref{prop:comp_stat1}, if there is an increasing differences change in the payoffs of a set $C$ of firms firms the all firms $j > \max{C}$ receive smaller matching sets. 
\end{lemma}

Proposition \ref{prop:comp_stat1} does not specify what happens to the matching sets for higher types than the one for which there is an increasing differences change. This will in general depend on the function $h$. This is because it is not clear what the effect of the described changes is on the value of the original matching sets for such types. The fact that all $j > k$ receive smaller matching sets benefits types $m < k$. However these types are also hurt by the fact that $k$ receives a larger matching set.

If $U^F(v,\cdot)$ is affine then comparative statics are much easier to identify. In this case, we can write $U^F(v,x) = \alpha^F(v) + \beta^F(v)\cdot x$ In this case, we can look at the problem entirely from the point of view of an individual, and the problem separates across individuals. In other words, the objective function can be written as

\begin{equation*}
\int_I\left( U^I(v^I_i, |s_I(i)|) + h(|s_I(i)|) \cdot \int_{s_I(i)} \beta^F(v^F_j) d \lambda(j) \right)d\lambda(i). 
\end{equation*}

An increasing differences change in payoffs here corresponds to an increase in $\beta^F(v)$. The following Lemma is immediate.

\begin{lemma}\label{lem:affine_welfare}
If $U_F(v,\cdot)$ is affine then the size of firm $j$'s matching set increases following an increasing differences change in payoffs. 
\end{lemma}

In this case we can also identify what happens when there is an increasing differences change to the payoffs of the lowest type. The following result is immediate given the threshold structure of matching sets. 

\begin{lemma}\label{lem:compstat_affine}
Suppose there is an increasing differences change in the payoffs of firm $N$, and that the supermodular order remains unchanged. If $U_F(v,\cdot)$ is affine then either firm $N$ is added to an individual's matching set or the individual's matching set remains unchanged. 
\end{lemma} 

\subsection{Individuals with private information}\label{sec:privateinfo}

Propositions \ref{prop1.1} and \ref{prop:comp_stat1} will be particularly interesting when there is a continuum of individuals who have private information about their type which must be elicited by the platform. Assume that Supermodularity holds on side $I$ (the arguments extend directly to Order-supermodularity). In this setting the usual \cite{myerson1981optimal} argument implies that in any incentive compatible mechanism the payoff of individual $i$ can be written as
\begin{equation}\label{eq:envelope}
V(v^I_i) = V(\und{v}^I) + \int\limits_{\und{v}^I}^{v^I_i} U^I_1(v, |s_I(v)|) dv.
\end{equation}
Suppose that the conditions of Proposition \ref{prop:comp_stat1} are satisfied, and that there is a increasing differences change in the payoffs of firm $1$. Corollary \ref{cor1.1} says that all individuals with types above $v^*_I(1)$ receive smaller matching sets. Under Supermodularity, $U^F_1(v,x)$ is increasing in $x$. Then if $V(v^*_I(1))$ does not increase, the envelope condition in (\ref{eq:envelope}) implies that all individuals will be worse off. This will be the case if $v^*_I(1) = \bar{v}^I$, which in turn will hold whenever both $U^F(v^F_1, \cdot)$ and $U^I(\und{v},\cdot)$ (or the variant of payoffs used in the platforms problem, if it is not the same as (\ref{eq:objective})) are increasing.

\begin{lemma}\label{lem:welfare}
When individuals' types are private information and $v^*_I(1)= \bar{v}^I$, all individuals are worse off when all receive smaller matching sets, and better off when all receive larger matching sets. 
\end{lemma}

In order to apply Lemma \ref{lem:welfare} it need not be the case that the platform's objective exactly corresponds to that in (\ref{eq:objective}). So long as the platform objective satisfies the conditions of Proposition \ref{prop:comp_stat1}. 

In many settings individual payoffs are quasi-linear and the platform wants to maximize transfers from individuals. Normalize the outside option of the lowest type, $U(\und{v}, 0)$, to zero. If individual payoffs are given by $u^I(v,|s_I(v)|) - t(v)$ then the sum of transfers from all individuals is given by 
\begin{equation}\label{eq:payments}
\int\limits_{\und{v}}^{\bar{v}}\left[ u^I(v,|s_I(v)|) - \dfrac{1- Q_F(v)}{q_F(v)} u_1^I(v, |s(v)|) \right] q_F(v) dv.
\end{equation}
Under Supermodularity, the matching mechanism is incentive compatible for individuals if and only if individual match sizes are increasing in types and payments are constructed to satisfy the envelope condition in (\ref{eq:envelope}). Proposition 1 will hold provided the integrand in (\ref{eq:payments}) is Supermodular. Moreover, under this condition the monotonicity constraint will not bind.\footnote{Without this condition Proposition \ref{prop1.1} will hold, i.e. the optimal matching will still have a threshold structure, but the proof of Proposition \ref{prop:comp_stat1} does not go through.}


\section{Extension: preferences over vertical types}\label{sec:extension}

One extension of the model, which arises naturally in many applications, is that individuals and/or firms may care about the vertical types of their matches. In the MVPD example explored above, we might think that channels prefer individuals with high viewing propensity, or individuals prefer channels with high quality programming (which in this case we interpret as high $v_F$). The main results will continue to hold when agents on one side prefer matches with higher type agents on the other side. This reinforces the optimality of giving better quality match sets to high type individuals, which is the key driver of the threshold structure, and thus the results that follow.\footnote{We could interpret the model in which agents care about the vertical types of those on the other side as a special case of the model presented above, in which the horizontal types are perfectly correlated with the vertical types. However under this interpretation the results of the previous section become trivial: a threshold structure of matchings conditional or horizontal type is meaningless when there is a single vertical type for each horizontal type. Moreover, we will be intersted in settings in which vertical types are private information. The case of privatly known horizontal types was not studied above.} 

Another natural extension is if the endogenous salience of an individual depends on the vertical types of the firms they are matched with. Again, if the endogenous salience is increasing in the types of an individual's matches then there will be no complications to the previous results. However in many cases this is not the direction we expect. An example, which I will explore in detail later on, is that monopolistically competitive firms may care not only about how many other firms their customers have access to, but also about whether these are high or low cost firms. Low cost (high type) are more damaging competitors because they are able to charge a lower price. Nonetheless, the general welfare conclusions will continue to hold; broadly speaking, increasing differences payoff changes for high type firms will hurt individuals, while such changes for low type firms will benefit individuals. 

In order to gain tractability when this holds, I assume that $U^F$ is affine in match quality. As we will see, this does not immediatly imply that the platform's problem will separted accross individuals, as would be the case if veritcal types were not payoff relevant for agents on the other side of the market; when agents' types are their private information the monotonicity constraint for incentive compatible mechanisms may bind. In this section I will provide conditions under which the problem is separable, characterize the optimal matching, and explore comparative statics. These results will be useful in applications, as illustrated in section \ref{sec:monopcomp}.

For simplicity, consider the model without horiontal differentiation. The payoffs of a type $v_I$ individual are given by $U^I(v_I,V_I(s_I))$ where 
\begin{equation*} 
V_I(s_I) = \int\limits_{s_I} v_j d\lambda(j).	
\end{equation*}
$U^I$ is strictly increasing in its second argument. 

The payoffs of a type $v_F$ firm are given by $\beta^F(v_F)\cdot V_F(s_F|S_I)$, where
\begin{equation*} 
V_F(s_F|S_I) = \int\limits_{s_F} h(v_I(i), V_I(s_I(i))) d\lambda(i)
\end{equation*}
and $h \geq 0$.

The platform's objective can be written as
\begin{equation}\label{eq:general_obj}
\int\left[ U^I(v_I, V_I(s_I(v_I))) + h(v_I,V_I(s_I(v_I))) \cdot \int\limits_{s_I(v_I)} \beta(v_F) dQ_F(v_F)\right]dQ_I(v_I).
\end{equation}
I will refer to maximizing the integrand in (\ref{eq:general_obj}) separatly for each individual as maximizing pointwise. Denote the integrand by $\Pi(s_I|v_I)$. The solution to the platform's problem may not be given by pointwise maximization, in particular if they are subject to a monotonicity constraint imposed by inceitnive compatability. We cannot appeal to Proposition \ref{prop1.1} to establish the threshold structure of optimal matchings in this setting, which would imply that the monotonicity constraint does not bind. This is because the proof of Lemma \ref{lem:monotonicity}, which says that higher type firms receive larger matching sets, does not go through. The reason is that higher type firms impose a larger negative externality on other firms by setting lower prices. However if the problem can be solved pointwise the solution is easy to characterize. This characterization also reveals the relevant assumptions needed to guarantee that the matchings have a threshold structure in vertical types. 

\begin{lemma}\label{lem:general_threshold}
Pointwise maximization of (\ref{eq:general_obj}) yeilds matchings for each individual type $v_I$ that are characterized by a threshold in $\beta(v_F)/v_F$: a firm is included if and only if $\beta(v_F)/v_F$ is high enough.
\end{lemma}
\begin{proof}
The integrand in (\ref{eq:general_obj}) only depends on the matching through the terms $V_I(s_I(v_I))$ and $\int_{s_I(v_I)} \beta(v_F)dQ_F(v_F)$. Consider the problem of maximizing $\int_{s_I(v_I)} \beta(v_F)dQ_F(v_F)$ subject to $V_I(s_I(v_I)) = \bar{V}$. Since the objective and the constraint are linear, this amounts to including a firm in the matching set if and only if $\beta(v_F)/v_F$ is high enough. 
\end{proof}

In immediate implication of Lemma \ref{lem:general_threshold} is that $v_F \mapsto V_F(s_F|S_I)$ is increasing in $\beta(v_F)/v_F$. However it says nothing about the comparison of matching sets for different individuals. Monotonicity properties of individual side matchings can be derived from the usual conditions (single crossing, interval order dominance, etc.) on the integrand. I will explore an application in detail, and delay further discussion of such results until then. Instead, I will present a useful comparative statics result on the matching set of a given individual that makes use of the threshold structure indentified in Lemma \ref{lem:general_threshold}. As we will see, it will be intersting to compare the optimal matchings under differnt firm payoffs, captured by different funtions $\beta$ and $\tilde{\beta}$.

Assume $v \mapsto \beta(v)/v$ and $v \mapsto \tilde{\beta}(v)/v$ are increasing. This is not necessary to obtain the types of results that follow, but it simplifies the notation, and will, in any case, be satisfied in many applications.\footnote{If these functions are not increasing then we can attain similar comparative statics results with respect to the order they induce.} If this holds then each individual's optimal matching set will be defined by a threshold in firm type. Let $\Pi(s_I|v_I,\beta)$ and be the integrand of $\ref{eq:general_obj}$ under $\beta$. Let $v^*(s_I)$ and $\tilde{v}^*(s_I)$ be the threshold types that maximize $\Pi(s_I|v_I,\beta)$ and $\Pi(s_I|v_I,\tilde{\beta})$ respectively; these define the optimal matchings under pointwise maximization. The following comparative statics result is straightforward, but will be useful in applications. 

\begin{lemma}\label{lem:B_compstat_simple}
Let $h(s_I,\cdot)$ be strictly decreasing. Assume $v \mapsto \beta(v)/v$ and $v \mapsto \tilde{\beta}(v)/v$ are increasing. If $\tilde{\beta}(v) = \beta(v)$ for all $v \leq v^*(s_I)$ then $\tilde{v}^*(s_I) \geq v^*(s_I)$. If $\tilde{\beta}(v) = \beta(v)$ for all $v \leq v^*(s_I) + \varepsilon$ for some $\varepsilon > 0$ and $\tilde{\beta}(v) \geq \beta(v)$ for all $v > v^*(s_I) + \varepsilon$, with strict inequality on a positive measure set, then $\tilde{v}^*(s_I) > v^*(s_I)$.
\end{lemma}
\begin{proof}
Abusing notation, for a matching set $s_I$ defined by a threshold $v^*$, denote the integrand by $\Pi(v^*|v_I,\beta)$. If $\tilde{\beta}(v) = \beta(v)$ for all $v \leq v^*(s_I)$, for any $v^* < v^*(s_I)$, we have $\Pi(v^*|v_I,\tilde{\beta}) - \Pi(v^*(s_I)|v_I,\tilde{\beta}) = \Pi(v^*|v_I,\beta) - \Pi(v^*(s_I)|v_I,\beta)$. This implies that $\tilde{v}^*(s_I) \geq v^*(s_I)$. 

The second part of the lemma, the strict inequality, follows from the fact that $h$ is strictly decreasing and $\int_{v^*}^{\bar{v}} \beta(v) dQ_F(v) > \int_{v^*}^{\bar{v}} \beta(v) dQ_F(v)$ under the stated assumptions. 
\end{proof}

On the other hand, if the increase in $\beta$ occurs for firms below $v^*(s_I$ than the oposite conclusion will hold. The proof is symmetric to that of Lemma \ref{lem:B_compstat_simple}.

\begin{lemma}\label{lem:B_compstat_simple2}
Let $h(s_I,\cdot)$ be strictly decreasing. Assume $v \mapsto \beta(v)/v$ and $v \mapsto \tilde{\beta}(v)/v$ are increasing. If $\tilde{\beta}(v) = \beta(v)$ for all $v > v^*(s_I)$ then $\tilde{v}^*(s_I) \leq v^*(s_I)$. If $\tilde{\beta}(v) = \beta(v)$ for all $v \geq v^*(s_I) - \varepsilon$ for some $\varepsilon > 0$ and $\tilde{\beta}(v) \geq \beta(v)$ for all $v < v^*(s_I) - \varepsilon$, with strict inequality on a positive measure set, then $\tilde{v}^*(s_I) < v^*(s_I)$.
\end{lemma}

Suppose that there is no direct individual-side component to platform payoffs, so $U^I = 0$. For example, an online retail platform that generates revenue by charging fees to sellers, but which allows customers free access to the site. Under this assumption we can draw stronger comparative statics conclusions. Proposition \ref{prop:B_compstat} is relevant for considering the effects of technological change on the matches. 

\begin{proposition}\label{prop:B_compstat}
Assume $h(v_I,\cdot)$ is decreasing and differentiable for all $v_I$, $U^I = 0$, $\beta(v)/v$ is increasing, and $\int_{\und{v}}^{\bar{v}} \beta(v)dv \geq 0$.\footnote{A similar result holds if $\beta(v)/v$ is not increasing, we just have to define $\alpha$ to be increasing in the order on types induced by $\beta(v)/v$.} Suppose there exists an increasing and strictly positive function $\alpha$ such $\tilde{\beta}(v) = \alpha(v) \beta(v)$. Then under pointwise maximization all individuals and all firms recieve smaller matching sets under $\tilde{\beta}$ than under $\beta.$
\end{proposition}

The proof of Proposition \ref{prop:B_compstat} makes use of the following result (see \cite{quah2009comparative} for a proof).
\begin{lemma}\label{lem:integral_ineq}
Suppose $[x',x'']$ is a compact interval or $\mathbb{R}$ and that $\alpha$ and $k$ are real valued functions on $[x',x'']$, with $k$ integrable and $\alpha$ increasing (and thus integrable as well). If $\int_x^{x''} h(t)dt \geq 0$ for all $x \in [x',x'']$ then 
\begin{equation*} 
\int\limits_{x'}^{x''} \alpha(t)h(t)dt \geq \alpha(x')\int\limits_{x'}^{x''}h(t)dt. 
\end{equation*}
\end{lemma}

Proposition \ref{prop:B_compstat} follows from Proposition 2 and Theorem 1 of \cite{quah2009comparative}, and Lemma \ref{lem:integral_ineq}.

\begin{proof}\textit{Proposition \ref{prop:B_compstat}}.
Consider an idividual with type $v_I$. Lemma \ref{lem:general_threshold} means that the individual's matching set will be defined by a threshold firm type $v^*$. Then, abusing the definition of $V_I$, we can write
\begin{equation*} 
\Pi(v^*|\tilde{\beta}) = h(v_I, V_I(v^*)) \cdot \int\limits_{v^*}^{\bar{v}_F} \tilde{\beta}(v) q_F(v) dv 
\end{equation*}
where $V_I(v_*) = \int_{v^*}^{\bar{v}_F} v q_F(v)dv$. The derivative of the objective with respect to $v^*$ is
\begin{align*} 
\Pi'(v^*|\tilde{\beta}) &= -h_2(v_I,V_I(v^*))v^*q_F(v^*) \cdot \int\limits_{v^*}^{\bar{v}_F} \tilde{\beta}(v) q_F(v) dv - h(v_I,V_I(v^*)) \cdot \tilde{\beta}(v^*)q_F(v^*) \\
& \geq -h_2(v_I,V_I(v^*))v^*q_F(v^*)\alpha(v^*) \cdot \int\limits_{v^*}^{\bar{v}_F} \tilde{\beta}(v) q_F(v) dv - h(v_I,V_I(v^*)) \cdot\alpha(v^*) \tilde{\beta}(v^*)q_F(v^*) \\
& = \alpha(v^*)\Pi'(v^*|\beta).
\end{align*}
Proposition 2 and Theorem 1 of of \cite{quah2009comparative} then implies that the optimal threshold is higher for $\tilde{\beta}$ than $\beta$, which means that all individuals recieve smaller matching sets. Then all firms also recieve smaller matching sets.
\end{proof}

\section{Applications}\label{sec:applications}

\subsection{Cable packages}

A monopolistic multi-channel video program distributor (MVPD), such as DirectTV or Comcast, faces a population of viewers with unknown values for programming. The MVPD offers a menu of packages of different channels to viewers. At the same time, the MVPD negotiates carriage fees with channels (referred to as video programmers in the IO literature). Video programmers benefit from viewers through advertising revenue, but may differ in terms of their cost of producing programming or their attractiveness to advertisers. From a programmers perspective, a viewer who has access to a large number of additional channels is less valuable, as they are likely to spend less time watching a given channel. The objective of the MVPD is to maximize revenue. 

The payoff of individual $i$ with type $v^I_i$ who purchases package consisting of a set $s$ or channels at a price of $p$ is given by $v^I_i \cdot g_I(|s|) - p$, where $g_I$ is increasing and $g_I(0) = 0$. Implicit in this formalization is the assumption that viewers like all channels equally, and only care about the number of channels they have access to. I will discuss this assumption later on. 

As discussed above, the total revenue to the platform from viewer fees generated by a direct mechanism $s: [\und{v}_I, \bar{v}^v] \mapsto 2^F$, where $A$ is the set of channels, is given by 
\begin{equation}\label{eq:revenue}
\int\limits_{\und{v}_I}^{\bar{v}_I} \left(v - \dfrac{1- Q_I(v)}{q_I(v)} \right) g_I(|s(v)|) dQ_I(v).
\end{equation}
Assume $Q_I$ is regular in the sense of \cite{myerson1981optimal}, so that the virtual values $\varphi(v) = v - (1-Q_I(v))/q_I(v)$ are increasing. Lemma \ref{lem:welfare} applies in this setting. After showing that the conditions of Propositions \ref{prop1.1} and \ref{prop:comp_stat1} are satisfied, we will identify changes to the environment which make all individuals worse off. 

The payoffs of channel $j$ matched with a set $s_F(j)$ of viewers is $U^F(v^j, |s_F(j)|_{S_I})$. I assume that the function $h$ determining individuals' endogenous salience is decreasing; the more channels an individual has access to the less valuable they are. Assume that $U^F$ is supermodular and concave in its second argument. 

The platform bargains with channels over a payment to be made for the right to carry their programming. In this industry, negotiations are generally over a monthly per subscriber ``affiliate fee'' that the MVPD pays the channel for every subscriber who has access the channel, whether the subscriber watches it or not. Assume that the platform is able to commit to a set of cable packages it will offer before negotiations over affiliate fees take place. The outcome of the multilateral bargaining game between the MVPD and channels is described by the Nash-in-Nash solution concept. Assume that if negotiations between a channel and the MVPD break down then the channel receives a payoff of zero. In this scenario the MVPD simply drops this channel from the existing packages. This results in a new allocation for consumers. If the original allocation had a threshold structure then this new allocation remains implementable: monotonicity is preserved when when one of the channels is dropped, holding fixed the remaining channels offered to each type. The total number of consumers who have access to each channel does not change, and so the MVPD's revenue from affiliate fees from other channels does not change. Thus the only effect of such a breakdown is the change it induces on the payments made by individuals for their packages. Thus we can think of negotiations as being simply over the total payment made from the MVPD to the channel, given the menu of bundles it plans to offer individuals. The price negotiated with channel $j$ is given by
\begin{align*}
p^*_j &= \arg\max_{p} \left( p + R(s) - R^O_j(s) \right)^{\beta}(U^F(v^F_j, |s_F(j)|_{S_I})) \\
& = \beta U^F(v^F_j, |s_F(j)|_{S_I}) + (1-\beta)(R^O_j(s) - R(s))
\end{align*}
where $R(s)$ is the individual-side revenue given in (\ref{eq:revenue}) and $R^O_j(s)$ is the individual-side revenue if channel $j$ is dropped. The difference between the two is given by 
\begin{equation*}
R^O(s_F(j)) - R(s) = \int\limits_{s_F(j)} \varphi(v) \cdot\left[g_I(|s_I(v)| - 1) - g_I(|s_I(v)|) \right]dQ_I(v).
\end{equation*}
Using the expression for $p^*_j$, we can write the MVPD firm-side revenue as 
a function of the matching as 
\begin{align*}
\sum_j p^*_j &= \beta \sum_j U^F(v^F_j, |s_F(j)|_{S_I}) + (1-\beta) \sum_j \int\limits_{s_F(j)} \varphi(v) \cdot\left[g_I(|s_I(v)| - 1) - g_I(|s_I(v)|) \right]dQ_I(v)\\
& = \beta \sum_j U^F(v^F_j, |s_F(j)|_{S_I}) + (1-\beta)\int\limits_{v^*_I(1|s)}^{\bar{v}_I} \varphi(v)|s_I(v)|\cdot \left[g_I(|s_I(v)| - 1) - g_I(|s_I(v)|) \right]dQ_I(v)
\end{align*}
where $v^*_I(1|s)$ is the threshold type for firm $1$'s matching set, which is also the lowest type individual that purchases a non-empty package. Then total MVPD revenue is given by
\begin{equation}\label{eq:revenue2}
\begin{split}
\beta \sum_j U^F(&v^F_j, |s_F(j)|_{S_I}) \\
&+ \int_{\und{v}_I}^{\bar{v}_I} \varphi(v)\big( g_I(|s_I(v)|) + (1-\beta)|s_I(v)|\cdot \left[g_I(|s_I(v)| - 1) - g_I(|s_I(v)|) \right]\big)dQ_I(v).
\end{split}
\end{equation}
It is worth pointing out that if $g_I(x) = ax + b$  then the MVPD's objective simplifies to 
\begin{equation*}
\beta \sum_j U^F(v^F_j, |s_F(j)|_{S_I}) + \int\limits_{v^*_I(1|s)}^{\bar{v}_I} \varphi(v)\big( \beta a \cdot |s_I(v)| + b\big)dQ_I(v).
\end{equation*}
In the special case of $a = 1$ and $b = 1$ maximizing this objective is the same as maximizing total (second-best) welfare, given by
\begin{equation*}
\sum_j U^F(v^F_j, |s_F(j)|_{S_I}) + \int\limits_{\und{v}_I}^{\bar{v}_I} \varphi(v)\cdot g_I(|s_I(v)|) dQ_I(v).
\end{equation*}

The general form of the objective in (\ref{eq:revenue2}) will satisfy the conditions of Propositions \ref{prop1.1} and \ref{prop:comp_stat1} if $U^F(v,\cdot)$ is concave and the integrand in the second term of (\ref{eq:revenue2}) is supermodular. This second requirement will be satisfied iff $g_I(x) + (1-\beta)x\cdot \left[g_I(x - 1) - g_I(x) \right]$ is increasing. A sufficient condition for this, given that $g_I$ is increasing, is that $g_I$ is concave.
\begin{lemma}\label{lem6}
If $g_I$ is concave then $g_I(x) + (1-\beta)x\cdot \left[g_I(x - 1) - g_I(x) \right]$ is increasing for $x \geq 1$.
\end{lemma}
\begin{proof}
For $x'' > x'$ we want to show 
\begin{equation*}
g_I(x'') - g_I(x') + (1-\beta)x''\cdot \left[g_I(x'' - 1) - g_I(x'') \right]  - (1-\beta)x'\cdot \left[g_I(x' - 1) - g_I(x') \right] \geq 0.
\end{equation*}

By concavity, $g_I(x' - 1) - g_I(x') \geq g_I(x'' - 1) - g_I(x'') \geq 0$. If $x''\cdot \left[g_I(x'' - 1) - g_I(x'') \right] \geq x'\cdot \left[g_I(x' - 1) - g_I(x') \right]$ then we are done, so suppose this does not hold. Under this assumption, since $(1-\beta) < 1$ it suffices to show  
\begin{equation*}
g_I(x'') - g_I(x') + x''\cdot \left[g_I(x'' - 1) - g_I(x'') \right]  - x'\cdot \left[g_I(x' - 1) - g_I(x') \right] \geq 0.
\end{equation*}
Adding and subtracting $x'[g_I(x''- 1) - g_I(x')]$, we can rewrite the left hand side as
\begin{equation*}
g_I(x'') - g_I(x') + (x''- x')\cdot \big[g_I(x'' - 1) - g_I(x'') \big]  - x'\cdot \big[g_I(x' - 1) - g_I(x') - [g_I(x'' - 1) - g_I(x'')]\big]
\end{equation*}
If $g_I$ is concave, $g_I(x'') - g_I(x'' - 1) \leq [g_I(x'' - 1) - g_I(x' -1)]/(x'' - x')$. Using this inequality, we obtain
\begin{align*}
&g_I(x'') - g_I(x') + (x''- x')\cdot \big[g_I(x'' - 1) - g_I(x'') \big]  - x'\cdot \big[g_I(x' - 1) - g_I(x') - [g_I(x'' - 1) - g_I(x'')]\big] \\
&\geq g_I(x'') - g(x'' - 1) - [g_I(x') - g(x'-1)] - x'\cdot \big[g_I(x' - 1) - g_I(x') - [g_I(x'' - 1) - g_I(x'')]\big]\\
&= (1-x')[g_I(x'-1) - g_I(x')] - (1 - x')\cdot[g_I(x'' - 1) - g_I(x'')]\\
& \geq 0 
\end{align*}
where the first inequality follows from $g_I(x'') - g_I(x'' - 1) \leq [g_I(x'' - 1) - g_I(x' -1)]/(x'' - x')$, and the final inequality from $x' \geq 1$ and concavity of $g_I$. 
\end{proof}

The discussion thus far of the MVPD problem is summarized in the following proposition

\begin{proposition}
Assume $g_I$ is concave and increasing and $U^F$ is supermodular and concave in its second argument. Then the MVPD's objective in (\ref{eq:revenue2}) satisfies the conditions of Propositions \ref{prop1.1} and \ref{prop:comp_stat1}.
\end{proposition}

If we assume that $U^F(v,\cdot)$ is increasing for all $v$ then a sufficient condition for no individuals to be excluded is that $\varphi(\und{v}_I) >0 $. If this holds then all individuals will be made worse off by increasing differences shifts in the payoffs of the highest type channels. 

The assumption that viewers care only about the number of channels they have access to may seem unnatural. High type channels as those which are best able to convert viewers who have access to their channel into profits, which are generated by add revenue. A natural interpretation is that these channels attract a larger portion of the available viewers then low type channels, but this interpretation suggests non-trivial viewer preferences. Viewers with the same ordinal preferences over channels can be accommodated, especially if the ordering corresponds to the ordering on channel types. However if viewers have different ordinal preferences over channels then threshold matchings may not be optimal. None the less, conditional on the MVPD choosing to offer threshold matchings all viewers will prefer larger matching sets. In reality, MVPDs almost always offer a menu of nested bundles. This may be because there is uncertainty about which programs channels will offer, so that viewers do not in fact have strong preferences over channels ex-ante, although they may develop such preferences after purchasing a bundle. It could also be because the distribution of preferences conditional on vertical types $v_I$ is such that profitable screening along this dimension is not possible. In this sense, the model seems to be a reasonable approximation of reality. 

I will now discuss situations in which increasing differences changes in firm payoffs make all individuals worse off.

\vspace{3mm}
\noindent \textit{Changes in viewing patters.} The most obvious cause of an increasing differences change in firm payoffs is technological change - direct shifts in firm payoffs. In this context, technological change could mean better programming or changes in viewing patterns that make some channels relatively more attractive to advertisers than others. A ``high type'' channel in this setting is one that most effectively converts viewers who have access to the channel into add revenue. This could be because of the demographic composition of their viewers or the broadness of their appeal. High type channels can be identified as those offered in the basic cable packages, such as CBS and FOX, whereas low type channels have more niche appeal, such as Animal Planet or HBO. Suppose viewers become less interested in watching the nightly news on NBC, and more interested in watching serial shows such as those offered by HBO. This can be interpreted as an increasing differences change in the payoffs of high type channels relative to NBC. Assuming $U_F$ is affine in its second argument, by Lemmas \ref{lem:affine_welfare} and \ref{lem:welfare}, we would expect such a change to lead to more packages that include HBO, with prices adjusting in a way that makes consumers better off. 

\vspace{3mm}
\noindent \textit{Horizontal mergers.} Suppose there is a merger between two channels. Mergers are often executed because of ``synergies'' between the firms involved. Moreover, firms will often defend the proposed merger against anti-trust challenges by arguing that synergies, in particular those that reduce costs, will benefit consumers. 

As an example to illustrate mergers with cost synergies in the context of this model, suppose that given a matching set $s$, a channel chooses the quality $q$ of its programming to maximize $(r(q) - c(v,q)) \cdot |s|_{S_I}$. Here $|s|_{S_I}$ is interpreted as the channels viewership potential, $r(q)$ is the realized add revenue per potential viewer, and $c$ is the cost per-viewer of producing programming.\footnote{An alternative formulation for channel profits would be $r(q)g(\cdot|s|_{S_I}) -c(v,q)$, i.e. zero marginal costs. In this case however the value function is supermodular, but it will only be concave if $g$ is sufficiently concave.} Assume $r$ is increasing and $c$ is submodular (decreasing differences in $v,q$). The firm's optimal choice of program quality is independent of its viewership, and increasing in $v$, so $U^F(v,x) = \max_q (r(q) - c(v,q)) \cdot |s|_{S_I}$ is supermodular and linear in its second argument. 

Cost synergies entail a reduction in the marginal cost of production for the firms involved. In this context, suppose firms with values $v''$ and $v'$ merge. Because of cost synergies, they will each have new cost functions $\hat{c}$ that are point-wise lower than their original costs. Then their margins will increase, i.e. $\max_q r(q) - \hat{c}(v,q)$ will be larger. This is an increasing differences shift in their payoffs. Suppose there is a merger between the lowest-value channels. Then Lemma \ref{lem:compstat_multi} implies that all cable packages will become larger. If on the other hand the two highest value channels merge then all cable packages will shrink. If there was no exclusion of individuals to begin with then all individuals will be made worse off. One could argue that this welfare analysis is incomplete, since the quality of programming increases following the merger. The positive conclusion regarding mergers of low value firms is unchanged by this consideration. In the case of high value firm mergers, the net effect will depend on how much individuals value higher quality programming. 

\vspace{3mm}
\noindent \textit{Vertical mergers.} Suppose the MVPD purchases a channel. This has two effects on its matching problem. One is that the MVPD will now appropriates the entire surplus generated by the purchased channel. This will have the effect of an increasing differences shift in this channel's payoffs. However the acquisition also affects the negotiations with other firms. If a channel fails to reach an agreement with the MVPD and is dropped from the packages, some viewers who would have had access to both this channel and the MVPD-owned channel will migrate to the latter. This benefits the MVPD, whereas before the merger the MVPD appropriated none of this additional benefit. This change in the outside option of the MVPD is termed the ``bargaining leverage over rivals'' (BLR) effect by \cite{Rogerson2019}. The effect on the structure of bundling will depend on the magnitude of the BLR effect.

If the MVPD purchases firm 1 then the price negotiated with channel $j$ is given by 
\begin{small}
\begin{align*}
\tilde{p}^*_j &= \arg\max_{p} \left(p + R(s) - R^O_j(s) + \theta \left(U^F(v^F_1, |s_F(1)|_{S_I}) - U^F(v^F_1, |s_F(1)|_{\tilde{S}^j_I})\right) \right)^{\beta}(U^F(v^F_j, |s_F(j)|_{S_I}))^{(1-\beta)} \\
& = \beta U^F(v^F_j, |s_F(j)|_{S_I}) + (1-\beta)\theta\left(R^O_j(s) - R(s) + U^F(v^F_1, |s_F(1)|_{\tilde{S}^j_I}) - U^F(v^F_1, |s_F(1)|_{S_I})  \right)
\end{align*}
\end{small}
where $|s_F(j)|_{\tilde{S}^j_I}$ is quality of channel 1's matching set when firm $j$ is dropped from all matchings. The parameter $\theta \in [0,1]$ captures the degree to which the MVPD internalizes the change in the viewership of the purchased channel should negotiations break down with channel $j$. If $\theta = 0$ then there is no MVPD effect. The MVPD's total revenue is 
\begin{small}
\begin{equation*}
\sum_{j=2}^N \tilde{p}^*_j + U^F(v^F_1, |s_F(1)|_{S_I}) + R(s) = \sum_{j = 2}^N p^*_j + (1- \beta)\theta \sum_{j=2}^N \left( U^F(v^F_1, |s_F(1)|_{\tilde{S}^j_I}) - U^F(v^F_1, |s_F(1)|_{S_I}) \right) + R(s) 
\end{equation*}
\end{small}
which can be re-written as 
\begin{equation}\label{eq:blr_revenue}
\begin{split}
\beta \sum_{j=2}^N U^F(&v^F_j, |s_F(j)|_{S_I}) + U^F(v^F_1, |s_F(1)|_{S^j_I}) \\
&+ \int_{\und{v}_I}^{\bar{v}_I} \varphi(v)\big( g_I(|s_I(v)|) + (1-\beta)|s_I(v)|\cdot \left[g_I(|s_I(v)| - 1) - g_I(|s_I(v)|) \right]\big)dQ_I(v) \\
&- (1-\beta)\int\limits_{s_F(1)} \varphi(v) \cdot\left[g_I(|s_I(v)| - 1) - g_I(|s_I(v)|) \right]dQ_I(v) \\
&+ (1- \beta)\theta \sum_{j=2}^N \left( U^F(v^F_1, |s_F(1)|_{\tilde{S}^j_I})- U^F(v^F_1, |s_F(1)|_{S_I}) \right).
\end{split}
\end{equation}
The first two lines of (\ref{eq:blr_revenue}) are the same as the original MVPD revenue in (\ref{eq:revenue2}), except that the weight on channel 1's revenue is now 1 instead of $\beta$. The third and fourth lines complicate the analysis. The third line reflects the fact that the individual side payoffs in the even of a breakdown in negotiations with firm 1 need no longer be considered. The fourth line is the BLR effect. Here payoffs do not have the same form as discussed above, and we can not directly apply or previous results. Under some additional assumptions however we can draw welfare conclusions.\footnote{Just assuming that $U^F(v,\cdot)$ is affine should be enough to do comparative statics, but this will require additional work.} 

First, assume that there is no BLR effect, so $\theta = 0$ and the fourth line of (\ref{eq:blr_revenue}) disappears. Second, suppose $g_I$ affine with slope $a$. Then the second line of (\ref{eq:blr_revenue}) reduces to 
\begin{equation*}
(1-\beta) a \int\limits_{s_F(1)} \varphi(v)dQ_I(v).
\end{equation*}
The only effect of this on the MVPD's problem is to increase the marginal benefit of adding viewers to the matching set of firm 1. Thus the result of Proposition \ref{prop:comp_stat1} will continue to hold. 

\begin{lemma}
Assume that there is no BLR effect, $g_I$ is affine, and $U^F(v, \cdot)$ is concave. Then if there is no exclusion of viewers and the MVPD purchases the highest type channel, all individuals will be worse off. If instead $U^F(v, \cdot)$ is affine and the MVPD purchases the lowest type channel then all individuals will be better off.  
\end{lemma}

The BLR effect reduces the price that the platform pays to all channels. However it is not immediately clear how this affects the marginal benefit of increasing bundle size. Analyzing the merger when the BLR effect is present is an interesting objective, but one which I leave for future work.


\subsection{Platform-mediated monopolistic competition}\label{sec:monopcomp}

Retailers, particularly on-line retailers such as Amazon, exert considerable control over the set of products to which consumers have access. This influence may take the form of placement of products in stores or in search results pages, advertising, or product recommendations, among others. Firms care about how many customers see their product and how many other products these customers see. Additionally, a firm's payoffs will depend on the prices offered by the competing firms to which its customers also have access. The retailer may exert varying levels of control over firm pricing decisions. A grocery store may have a high degree of control, whereas Amazon may have little direct say in the prices set by firms. 

I will consider a version of the classic \cite{dixit1977monopolistic} model of monopolistic competition. There are a continuum of firms and individuals. For simplicity, assume that individuals are identical. I will discuss later how individual heterogeneity can be accommodated. The driving force here will be the desire of the platform to screen firms that have different marginal costs, and thus value access to consumers differently. The platform cannot directly control the prices set by firms, although it may do so indirectly through the choice of the matching sets. This assumption fits best with a model of an online platform such as Amazon. I will show that platform payoffs in this setting reduce to the form studied in \ref{sec:extension}.

\subsubsection{Customer side}

I first describe the demand of a customer who is matched with a set $s$ of firms, where firm $j$ has set price $p(j)$. The individual chooses a quantity $q(j)$ to purchase from firm $j$ and how much money $m$ to hold. The individual's choice solves
\begin{equation}\label{eq:custobjective}
\max_{q(\cdot),m} = m + \dfrac{1}{\theta}\left( \int_s q(j)^{\frac{\sigma - 1}{\sigma}} dj \right)^{\frac{\theta \sigma}{\sigma -1}} \ \ \text{s.t.} \ \ \int_s p(j) q(j) dj + m \leq w
\end{equation}
where $w$ is the individual's wealth, $\sigma > 1$, $\theta \in (0,1)$ and $\sigma(1-\theta) > 1$.\footnote{Preferences of the form in (\ref{eq:custobjective}) have been used by \cite{bagwell2015trade}, \cite{helpman2010labour}, and \cite{helpman1989trade} among others.} Assume that wealth is high enough that individuals hold a positive amount of money. Define the price index for matching set $s$ as
\begin{equation*}
P(s) = \left(\int_s p(j)^{1-\sigma} dj\right)^{\frac{1}{1 - \sigma}}.
\end{equation*} 
Then demand for money can be written as 
\begin{equation*}
m(s) = w - P(s)^{\frac{\theta}{\theta - 1}} 
\end{equation*}
and demand for the product of firm $j$ is given by\footnote{Derivations can be found in Appendix \ref{sec:pmmc}}
\begin{equation*}
q(j|s) = p(j)^{-\sigma} \cdot P(s)^{\frac{\sigma(1 - \theta) - 1}{1 - \theta}}.
\end{equation*}
The indirect utility of an individual is given by 
\begin{equation*}
\dfrac{1-\theta}{\theta} P(s)^{\frac{\theta}{\theta -1}} + w.
\end{equation*}

\subsubsection{Firm side}
A firm matched with a set $s_F(j)$ of customers faces an aggregate demand given by 
\begin{equation*}
p(j)^{-\sigma} \cdot \int\limits_{s_F(j)} P(s_I(i))^{\frac{\sigma(1 - \theta) - 1}{1 - \theta}} di.
\end{equation*}
Define $h(s) \equiv  P(s_I(i))^{\frac{\sigma(1 - \theta) - 1}{1 - \theta}}$. Assume that marginal costs are constant for all firms, but differ across firms. Marginal costs are the firm's private information, which the platform will like to elicit. For simplicity, assume that firms have no fixed cost. \footnote{Adding fixed costs, even if they differ across firms, does not change anything as long as all firms want to continue operating.} The firm's price setting problem can be stated as
\begin{equation*}
\max_p \ p^{1-\sigma} \cdot \int\limits_{s_F(j)} h(i) di -  c_j p^{-\sigma} \cdot \int\limits_{s_F(j)} h(i) di.
\end{equation*}
The solution to this problem is to set $p(j) = \frac{\sigma}{\sigma - 1} c_j$. The fact that the price does not depend on the matching set or the prices of other firms is the main benefit of assuming CES utility and constant marginal costs. Define $v_j \equiv c_j^{1-\sigma}$ (recall that $\sigma > 1$). The firms profits from matching set $s_F(j)$ are then given by 
\begin{equation*}
v_j \gamma \cdot \int\limits_{s_F(j)} h(i)di
\end{equation*} 
where $\gamma = \left(\frac{\sigma}{\sigma-1}\right)^{1-\sigma} - \left(\frac{\sigma}{\sigma-1}\right)^{-\sigma} > 0$. 

Using the firm pricing decision, payoffs of an individual with matching set $s_I$ are given by
\begin{equation}\label{eq:custwelfare}
w + \left( \frac{1-\theta}{\theta} \right) \left( \frac{\sigma}{\sigma-1} \right)^{\frac{\theta}{\theta - 1}} \left( \int\limits_{s_I} v_j dj \right)^{\frac{\theta}{(\theta - 1)(1 - \sigma)}}.
\end{equation}  
This can be written as $w + g_I(s_I)$.

\subsubsection{The platform}

There are a number of platform objectives that can be entertained in this setting. I will assume that the goal of the platform is to maximize a weighted sum of customer surplus and revenue collected from firms. The platform may care about customer surplus because it wants to attract customers. Alternatively, if the platform was also screening on the individual side we have seen how the individual-side objective looks like welfare maximization, with virtual values replacing true values. I could also consider a platform that collects transaction fees from users, and thus wants to maximize a weighted sum of customer transactions and revenues from firms. The results discussed here would not change. 

Payoffs for the firm and individual take the form studied in section \ref{sec:extension}: the endogenous salience of individual $i$, given by $h(i)$, depends on the types of the firms in $s_I(i)$. Using the firm's pricing decision, we can rewrite $P(s_I(i))$ in terms of firm types:
\begin{equation*}
h(s_I(i)) = P(s_I(i))^{\frac{\sigma(1 - \theta) - 1}{1 - \theta}} = \psi \cdot \left( \int\limits_{s_I(i)} v_j dj \right)^{\kappa}
\end{equation*}
where $\kappa =  \frac{\sigma(1-\theta) - 1}{(1-\theta)(1-\sigma)} \in (-1,0)$ and $\psi= \left(\frac{\sigma}{\sigma-1}\right)^{\kappa(1-\sigma)} > 0$.
The platform wants to screen firms based on marginal costs. Firm payoffs are of the form $v_j\cdot g_F(s_F(j))$, where $g_F$ is linear. Assuming that the distribution of $v_j$ is regular, the platform will solve for the optimal matching by replacing $v_j$ with $\varphi_F(v_j) = v_j - \frac{1 - Q_F(v_j)}{q_F(v_j)}$, subject to monotonicity of firm side match quality in type. 

Ignoring individual wealth, which is fixed, the platform payoff can be written as 

\begin{equation}\label{eq:amazonobj}
\int\limits_{\und{v}_I}^{\bar{v}_I} \left[g_I(s_I(v)) + \gamma h(s_I(v)) \cdot \int\limits_{s_I(v)} \varphi_F(v_j) dj\right]dQ_I(v).
\end{equation}
where $g_I(s_I) = \left(\frac{1-\theta}{\theta} \right) P(s_I)^{\frac{\theta}{\theta - 1}}$. The platform chooses $s_I(\cdot)$ to maximize (\ref{eq:amazonobj}), subject to monotonicity of firm match quality. If the platform cannot discriminate between customers, meaning that all customers must receive the same matching set, then monotonicity of the firm side matching implies that individual matching sets must be characterized by a threshold in $v_F$. If the platform is allowed to offer different matching sets to different individuals however, the problem is not necessarily separable across customers due to the firm-side monotonicity constraint.  Given the separable structure of the objective, Lemma \ref{lem:general_threshold} implies that the threshold structure will hold, under a slight strengthening of regularity.  

\begin{lemma}\label{lem:threshold_amazon}
Maximizing the integrand in (\ref{eq:amazonobj}), i.e. solving the problem separately for each individual, yields matchings characterized by a threshold in $\varphi_F(v)/v$; a firm with type $v$ is matched with the individual if and only if $\varphi_F(v)/v$ is high enough. 
\end{lemma}

Under the assumption that $\varphi_F$ is increasing, $\varphi_F(v)/v$ will always be increasing over the set of $v$ such that $\varphi_F(v) < 0$. This means that if we know that all firms with positive virtual values are matched with all individuals then under regularity we can conclude that the individuals' matching sets are characterized by a threshold in $v_F$. In general however we can only conclude this when $\varphi_F(v)/v$ is increasing.

\begin{corollary}
If $\varphi_F(v)/v$ is increasing then maximizing (\ref{eq:amazonobj}) pointwise yeilds matching sets characterized by a threshold in $v_F$.  
\end{corollary} 

If $\varphi(v)/v$ is increasing, or if the platform is not able to discriminate between individuals, then the problem is separable accross individuals, and can be solved by maximizing the integrand in (\ref{eq:amazonobj}). Lemma \ref{lem:threshold_amazon} implies that the solution to this problem is unique, so all customers will receive the same matching set, and we can talk about the ``representative customer''. 

Even when the problem is separable, comparative statics in this setting are complicated by the fact that firm types enter into both the endogenous salience of individuals and the profitability of firms. This is not an issue however if we consider changes that affect the virtual values without changing the true firm types, in which case we can apply Lemma \ref{lem:B_compstat_simple}. I will consider two such changes; if the platform receives more precise information about firm types, and if the platform purchases a subset of high type firms. The following Lemma is the immediate implication of \ref{lem:B_compstat_simple} in this setting. 

\begin{lemma}\label{lem:amazoncompstat}
Assume $\varphi_F(v)/v$ is increasing. Consider a set of firms, all of which are matched with the representative customer. If the virtual values of these firms increase then the representative customer will be worse off and receive a smaller matching set.
\end{lemma}

Suppose the platform receives some information about the vertical types of firms. For some partition $\tau$ of the set $[\und{v}_F, \bar{v}_F]$ of potential firm types the platform learns to which partition cell each firm belongs. The platform therfore only needs to worry about firms deviating to types that are in the same cell; the mechanism design problem on the firm side separates completely accross cells. Thus the virtual values are computed cell-by-cell; for a firm $v$ in cell $[v_{**}, v^{**}]$ the virtual value is $v - \frac{Q_F(v^{**}) - Q_F(v)}{q_F(v)}$. Moreover the monotonicity constraint is need only be satisfied within each cell. Say that a partition cell is \textit{included} if all firms in that cell are matched with the representative customer. The following are corollaries of Lemma \ref{lem:amazoncompstat}. 

\begin{corollary}\label{cor:amazonmerger}
Assume $\varphi_F(v)/v$ is increasing within each cell. If the platform purchases an included partition cell of firms then the customer recieves a smaller matching set and is worse off.
\end{corollary}
\begin{proof}
If the platform purchases an included partition cell it replaces the virtual values of each firm in this cell with the true values, which are higher. There is no change in the virtual values of firms in other cells.
\end{proof}
Assume $\varphi_F(v)/v$ is increasing. Corollary \ref{cor:amazonmerger} and Lemma \ref{lem:threshold_amazon} imply that if the inverse hazard rate is decreasing then the customer receives a smaller matching set. Even without this assumption, the merger makes the customer worse off, as the platform seeks to divert customers to the firms that it owns. Similarly, if the platform gets better information about the values of firms that are already included, it will be able to extract more surplus from these firms, and would thus like to divert customers to them. 

\begin{corollary}\label{cor:amazoninfo}
Assume $\varphi_F(v)/v$ is increasing within each cell. If the platform's partition over a set of included firms becomes finer then the customer recieves a smaller matching set and is worse off. 
\end{corollary}
\begin{proof}
If a given partition cell $[v_{**}, v^{**}]$ is subdivided then any firm in $[v_{**}, v^{**}]$ will be in a new cell with an upper bound that is below $v^{**}$, and strictly so for some firms. Then the virtual values of these firms will be higher.
\end{proof}

The natural counterpoints to $\ref{cor:amazonmergerhetero}$ and $\ref{cor:amazoninfohetero}$ also obtain; if the merger or information aquisition concerns cells of excluded firms, firms with which the customer is not matched, than the customer recieves larger matching sets and is better off

\subsubsection{Customer heterogeneity}

A natural dimension of customer heterogeneity is the degree to which customers value the goods sold on the platform relative to money. With quasi-liner preferences wealth heterogeneity does affect the purchase decisions. In some sense, the trade-off between goods and money is captured by the parameter $\theta$, so we can consider heterogeneity in this dimension. However higher $\theta$ does not exactly capture an intuitive notion of valuing goods relatively more. The quality of a customer's matching set is given by $\int_{s_I} v_j dj$. Customer preferences are supermodular in $\theta$ and match quality if and only if $P \geq 1$. In fact, if $P < 1$ it is easy to show that customer demand may be decreasing in $\theta$.\footnote{There is no technical problem with considering heterogeneity in $\theta$, although it does complicate the analysis when we consider the case in which the type is an individual's private information and there are transfers between the platform and individuals. The issue is one of interpretation.} Thus, I will consider an slight modification of customer preferences which admits a more natural interpretation of the individual type: let preferences for a type $v_I$ customer be represented by
\begin{equation*}
m + \frac{v_I^{1-\theta}}{\theta}\left( \int\limits_{s} q(j)^{\frac{\sigma - 1}{\sigma}} dj \right)^{\frac{\theta \sigma}{\sigma - 1}}
\end{equation*}
where $v_I$ is the individuals type. It is easy to show, following the same steps as before, that demand for good $j$ is given by 
\begin{equation*}
q(j|s,v_I) = v_I \cdot p(j)^{-\sigma}\cdot P(s)^{\frac{\sigma(1-\theta) - 1}{1-\theta}}
\end{equation*}
and customers' indirect utility is given by
\begin{equation*}
v_I \left(\frac{1-\theta}{\theta}\right) P(s)^{\frac{\theta}{\theta-1}}.
\end{equation*}
Define $h(s|v_I) \equiv v_I \cdot P(s)^{\frac{\sigma(1-\theta) - 1}{1-\theta}}$. The remainder of the derivations are as before.

If the platform can observe individuals' types, which may not be an unreasonable assumption for online platforms with access to detailed information about their customers, not much changes in the above analysis. Lemma \ref{lem:general_threshold} applies, so if $\varphi_F(v)/v$ is increasing then matchings for each individual have a threshold structure and the firm-side monotonicity constraint does not bind. Then the conclusions Corollaries \ref{cor:amazonmerger} and \ref{cor:amazoninfo} continue to hold for each individual: each individual is matched with a smaller set of firms if the platform purchases or gains more precise information about an included cell of firms. This makes individuals worse of if the platform is maximizing the weighted sum of customer welfare and firm-side revenue. Of course, if the platform can use transfers to extract surples from individuals then individuals recieve zero surpluss in either case.

We can also consider the situation in which an individual's type is their private information. If the platform can only control individual payoffs through the matching set, for example if monetary transfers are not possible, then again all individuals will receive the same matching set. Lemmas \ref{lem:threshold_amazon} and \ref{lem:amazoncompstat}; and Corollaries \ref{cor:amazonmerger} and \ref{cor:amazoninfo} continue to hold for all individuals. 

The primary case of interest is when individuals are privatly informed about their type, and transfers can be made between individuals and the platform. The necessary and sufficient conditions for customer-side incentive compatibility in this setting are monotonicity of match quality in type and payments that satisfy the envelope condition. Let $g_I$ be defined as above. The platform may still seek to maximize the sum of total individual-side welfare and net revenue, meaning the sum of transfers from firms and customers. Aggregate individual-side welfare can be written as 
\begin{equation*}\label{eq:amazon_welfareobj}
U(\und{v}_I) + \int\limits_{\und{v}_I}^{\bar{v}_I} g_I(s_I(v)) (1 - Q_I(v)) dv. 
\end{equation*} 
The platform does not benefit from transferring money to individuals since its payoffs are linear in money, so it is without loss to set $U(\und{v}_I) = 0$. The platform objective is to maximize
\begin{equation}\label{eq:amazonobjhetero}
\int\limits_{\und{v}_I}^{\bar{v}_I} \left[g_I(s_I(v)) (1 - Q_I(v)) + v \gamma h(s_I(v)) \cdot \int\limits_{s_I(v)} \varphi_F(v_j) dj\right]dQ_I(v)
\end{equation}
subject to monotonicity of firm and individual match qualities. Lemma \ref{lem:threshold_amazon} applies here, so if $\varphi_F(v)/v$ is increasing then maximizing the integrand in (\ref{eq:amazonobjhetero}) yields threshold matching sets for each individual. 

If we ignored the firm side, maximizing individual-side welfare would imply full ironing since $1-Q_I(v)$ is decreasing. All individuals would receive the same quality matching set, and thus the same set given that threshold matchings are optimal. In fact, the firm side reinforces this effect. The platform faces a trade-off between higher customer welfare and greater firm-side revenue. Since firm side revenue scales with $v_I$, while $1-Q_I(v_I)$ is decreasing, the platform will prioritize firm side revenue when dealing with higher type individuals. 

\begin{lemma}\label{lem:amazonpooling}
Assume $\varphi_F(v)/v$ is increasing. If the platform's objective is to maximize (\ref{eq:amazonobjhetero}), the sum of individual-side welfare and revenue, then there is full pooling: the optimal matching is the same for all individuals.
\end{lemma}
\begin{proof}
Suppose platform maximized its objective separately for each individual, ignoring the monotonicity constraint. I want to show that higher type individuals would get worse matching sets. The objective for an individual with type $v_I$ can be written as
\begin{equation*}
v_I\left(g_I(s_I)\frac{(1 - Q_I(v_I))}{v_I} + \gamma\psi \left(\int\limits_{s_I}  v_j dj\right)^{\kappa} \cdot \int\limits_{s_I} \varphi_F(v_j) dj\right).
 \end{equation*}
This objective satisfies single-crossing in $v_I$ and $- \int_{s_I} v_j dj$, so higher types get worse matching sets. Since point-wise maximization yields a decreasing allocation there is full ironing: the optimal matching subject to monotonicity has the same quality for all individuals. Since $\varphi_F(v)/v$ is increasing this implies that all matching sets will be the same. 
\end{proof}

Lemma \ref{lem:amazonpooling} is consistent with some observed patterns of platform structure. Early stage platforms, which are focused on attracting new users, often offer the same service to all customers. Later, when the user base has been established, do platforms begin do discriminate between customers. As we will see below, such discrimination arises when the platform seeks to translate individual-side users into revenue. 

Given Lemma \ref{lem:amazonpooling} we can again talk about the ``representative customer'' who in this case has a type given by the average type in the population. Lemma \ref{lem:amazoncompstat} and Corollaries \ref{cor:amazonmerger} and \ref{cor:amazoninfo} continue to apply. 

Finally, suppose the platform also wants to maximize firm-side and individual-side revenue. Individual side revenue can be written as 
\begin{equation*}\label{eq:amazon_revobj}
U(\und{v}_I) + \int\limits_{\und{v}_I}^{\bar{v}_I} \varphi_I(v)\cdot g_I(s_I(v)) dv 
\end{equation*}
where $\varphi_I(v) = v - \frac{1-Q_I(v)}{q_I(v)}$. The platform can extract full surpluss from the lowest type, so $U(\und{v}_I) = 0$. The platform objective is to maximize
\begin{equation}\label{eq:amazonrevobjhetero}
\int\limits_{\und{v}_I}^{\bar{v}_I} \left[\varphi_I(v)\cdot g_I(s_I(v)) + v \gamma h(s_I(v)) \cdot \int\limits_{s_I(v)} \varphi_F(v_j) dj\right]dQ_I(v)
\end{equation}
subject to individual and firm side monotonicity. We can rewrite the integrand as 
\begin{equation*}
v \left[\frac{\varphi_I(v)}{v} g_I(s_I(v)) + \gamma h(s_I(v)) \cdot \int\limits_{s_I(v)} \varphi_F(v_j) dj\right].
\end{equation*}
This objective satisfies single-crossing in $\varphi_I(v)/v$ and $\int_{s_I} v_j dj$, which has the following implication.

\begin{lemma}\label{lem:amazon_indivmon}
Point-wise maximization of total revenue, given by (\ref{eq:amazonrevobjhetero}), yields individual match qualities that are increasing in $\varphi_I(v)/v$. 
\end{lemma}

An immediate implication of Lemma \ref{lem:amazon_indivmon} is that point-wise maximization will yield monotone allocations if $\varphi_I(v)/v$ is increasing. This claim has a partial converse: if $\varphi_I(v)/v$ is not increasing then point-wise maximization will violate monotonicity, unless the violations of increasing $\varphi_I(v)/v$ happen to occur for individuals that are either excluded or fully matched. 

Assuming that both $\varphi_I(v)/v$ and $\varphi_F(v)/v$ are increasing. Then monotonoicity constraints do not bind for either individuals or firms. Thus we can apply Lemma \ref{lem:general_threshold}. and Corollaries \ref{cor:amazonmerger} and \ref{cor:amazoninfo} apply to each individual: each individual is matched with a smaller set of firms if the platform purchases or gains more precise information about an included cell of firms. Using the envelope expression for individual payoffs we have the following welfare conclusions.

\begin{corollary}\label{cor:amazonmergerhetero}
Assume $\varphi_F(v)/v$ is increasing within each cell. If the platform purchases a cell of firms that is included for all customers then all customers are worse off.
\end{corollary}

\begin{corollary}\label{cor:amazoninfohetero}
Assume $\varphi_F(v)/v$ is increasing within each cell. If the platform's partition over a set of included firms that are included for all customers becomes finer then all customers are worse off. 
\end{corollary}

\newpage

\section*{Appendix}
\appendix

\section{Platform-mediated monopolistic competition}\label{sec:pmmc}

The FOC for the customer's problem in (\ref{eq:custobjective}) gives
\begin{equation}\label{eq:custfoc}
p(k) = q(k)^{-\frac{1}{\sigma}} \left( \int_s q(j)^{\frac{\sigma - 1}{\sigma}} \right)^{\frac{\theta \sigma}{\sigma-1} - 1}.
\end{equation}
Multiplying both sides by $q(k)$ and integrating both sides over $s$ yields 
\begin{equation}\label{eq15}
\int_s p(j)q(j)dj = \left( \int_s q(j)^{\frac{\sigma - 1}{\sigma}} \right)^{\frac{\theta \sigma}{\sigma-1}}.
\end{equation}
We can also rearange (\ref{eq:custfoc}) to obtain 
\begin{equation}\label{eq16}
p(k)q(k) = p(k)^{1-\sigma} \left( \int_s q(j)^{\frac{\sigma - 1}{\sigma}} \right)^{\sigma\frac{\theta \sigma}{\sigma-1} - 1}.
\end{equation}
Integrating both sides over $s$ yields
\begin{equation}\label{eq17}
\int_s p(j)q(j)dj = P(s)^{1-\sigma} \left( \int_s q(j)^{\frac{\sigma - 1}{\sigma}} \right)^{\sigma\frac{\theta \sigma}{\sigma-1} - 1}.
\end{equation}
Combining (\ref{eq15}) and (\ref{eq17}) we obtain
\begin{equation*}
\int_s p(j)q(j)dj = P(s)^{\frac{\theta}{\theta-1}}
\end{equation*}
which is what allows us to write demands as a function of $P(s)$.

\newpage

\bibliography{matching}

\end{document}